\documentclass{article}
\usepackage{authblk}

\usepackage{hyperref}

\usepackage{amssymb}
\usepackage{amsmath}
\usepackage{amsthm}
\usepackage{enumitem}
\usepackage{xspace}
\usepackage{geometry}
\newtheorem{theorem}{Theorem}
\newtheorem{lemma}[theorem]{Lemma}
\newtheorem{definition}[theorem]{Definition}
\newtheorem{corollary[theorem]}{Corollary}

\newcommand{\nats}{\mathbb N}

\newcommand{\solnlabel}[1]{\label{sol:#1}}
\newcommand{\solnref}[1]{\ref{sol:#1}}
\newcommand{\blank}{\ensuremath{\mathtt{*}}}

\def\problem#1{{\scshape #1}}
\def\tarski{\problem{Tarski}\xspace}
\def\utarski{\problem{Unique-Tarski}\xspace}
\def\sutarski{\problem{Super-Unique-Tarski}\xspace}

\def\OPDC{\problem{OPDC}\xspace}
\newcommand{\up}{\ensuremath{\mathsf{up}}}
\newcommand{\down}{\ensuremath{\mathsf{down}}}
\newcommand{\zero}{\ensuremath{\mathsf{zero}}}

\DeclareMathOperator{\Up}{Up}
\DeclareMathOperator{\Down}{Down}

\title{Super Unique Tarski is in UEOPL}

\author[1]{John Fearnley}
\author[1,2]{Rahul Savani}
\affil[1]{University of Liverpool}
\affil[2]{Alan Turing Institute}
\date{}

\begin{document}

\maketitle

\begin{abstract}
We define the \sutarski problem, which is a \tarski instance in which all slices
are required to have a unique fixed point. We show that \sutarski lies in UEOPL
under promise-preserving reductions.
\end{abstract}

\section{Introduction}

A function $f: \{1, 2, \dots, n\}^k \rightarrow \{1, 2, \dots, n\}^k$ is
\emph{monotone} if $f(x) \le f(y)$ whenever $x \le y$.
Tarski's theorem~\cite{Tarski55} states that every monotone function has at
least one fixed point, which is a point $x$ such that $x = f(x)$. 
The \tarski problem is therefore defined to be the problem of finding a fixed
point of $f$. 

In this paper, we are interested in the computational complexity of the \tarski
problem. Prior work has shown that the problem lies in the complexity classes
PPAD and PLS~\cite{EPRY20}, which also means that the problem lies in
CLS~\cite{FearnleyGHS23}.

In particular, we are inspired by a recent result of Chen et
al.\ showing that the \tarski problem can be reduced, in the
black box model, to the \utarski problem~\cite{ChenLY23}, which captures \tarski instances that 
possess a unique fixed point. They explicitly ask whether \utarski can be shown to lie
in the complexity class UEOPL, which is a sub-class of CLS that was defined to
capture problems that have unique solutions~\cite{FGMS20}.

We do not resolve that problem here, but we do show that one further restriction
does allow us to show UEOPL containment. A \emph{slice} of a \tarski instance
is obtained by fixing some of the coordinates of the instance, and then
projecting $f$ onto the resulting sub-instance (this is defined formally in
Section~\ref{sec:sutarski}). Our contributions are as follows.

\begin{itemize}
\item We define the \sutarski problem, which captures \tarski instances that
have a unique fixed point, and in which \emph{every slice} of the instance also has a
unique fixed point.

\item We show that the \sutarski problem lies in UEOPL under promise-preserving
reductions.
\end{itemize}

\section{Definitions}

\paragraph{\bf The Tarski fixed point problem.}

We work with a complete lattice defined over a $k$-dimensional
grid of points, where each side of the grid has length $n \in \nats$. Formally,
we define our lattice to be $L^k = \{1, 2, \dots, n\}^k$. We use $\le$ to
denote the natural partial order over this lattice: we have that $x \le y$
if and only if $x, y \in L^k$ and $x_i \le y_i$ for all $i \le k$. 

A function $f : L^k \rightarrow L^k$ is \emph{monotone} if $f(x) \le
f(y)$ whenever $x \le y$. A point $x \in L^k$ is a \emph{fixed point} of $f$ if
$f(x) = x$. A fixed point $x$ of $f$ is the \emph{least} fixed
point if $x \le y$ for every other fixed point $y$ of $f$. Likewise, a fixed
point $x$ is the \emph{greatest} fixed point if $y \le x$ for every other
fixed point $y$ of $f$.
Tarski's fixed point theorem can be stated as follows.
\begin{theorem}[\cite{Tarski55}]
Every monotone function on a complete lattice has a greatest and least
fixed point.
\end{theorem}
In particular, this implies that every monotone function has at least
one fixed point.
Thus, we can define a total search problem for Tarski's fixed point theorem.
\begin{definition}[\tarski]
Given a function $f : L^k \rightarrow L^k$, find one of the following.
\begin{enumerate}[label=(T\arabic*), wide=0pt]
\item \solnlabel{T1} A point $x \in L^k$ such that $f(x) = x$.
\item \solnlabel{T2} Two points $x, y \in L^k$ such that $x \le y$ and $f(x) \not\le
f(y)$.
\end{enumerate}
\end{definition}
Solutions of type \solnref{T1} are fixed points of the function $f$, whereas
solutions 
of type \solnref{T2} are succinct witnesses that $f$ is not monotone.
By Tarski's fixed point theorem, if a function $f$ has no solutions of type
\solnref{T2}, then it must have a solution of type \solnref{T1}, and so \tarski
is a total problem, meaning that it always has a solution.

\section{The \sutarski Problem}
\label{sec:sutarski}

We now define our restriction on the \tarski problem.

\paragraph{\bf \tarski instances with unique fixed points.}

The proof of Tarski's fixed point theorem actually gives us some stronger
properties. Given a function $f$ over a lattice
$L^k$, we define the following two sets.
\begin{align*}
\Up(f) &= \{ x \in L^k \; : \; x \le f(x) \}, \\
\Down(f) &= \{ x \in L^k \; : \; f(x) \le x \}.
\end{align*}
We call $\Up(f)$, the \emph{up} set, which contains all points in which $f$
moves upward
according to the ordering~$\le$, and likewise we call $\Down(f)$ the
\emph{down} set. Note that the set of fixed points of $f$ is exactly $\Up(f) \cap
\Down(f)$. A stronger version of Tarski's fixed point theorem can be
stated as follows.
\begin{theorem}[\cite{Tarski55}]
\label{thm:tarski2}
Let $f$ be a monotone function, let $l$ be the least fixed point of
$f$, and let $g$ be the greatest fixed point of $f$. We have
\begin{itemize}
\item $l \le x$ for all $x \in \Down(f)$, and
\item $x \le g$ for all $x \in \Up(f)$.
\end{itemize}
\end{theorem}

\noindent Note that if $f$ has a \emph{unique} fixed point $p$, then $p$ is
simultaneously the greatest and least fixed point. The theorem above implies
that $p$ must lie exactly between the sets $\Up(f)$ and $\Down(f)$:
specifically $x \le p$ for all $x \in \Up(f)$, and $p \le x$ for all $x
\in \Down(f)$.

This gives us a succinct witness that a \tarski instance does not have a
unique fixed point, 
which has already been observed by Chen et al.~\cite{ChenLY23}.
\begin{lemma}[\cite{ChenLY23}]
\label{lem:non-unique}
Let $f$ be a monotone function. There is a unique fixed point of $f$ if
and only if there is no pair of points $x, y \in L^k$ such that $x \not\le y$
and $x \in \Up(f)$ and $y \in \Down(f)$.
\end{lemma}
\begin{proof}
If $f$ has a unique fixed point $p \in L^k$, then for any points $x \in \Up(f)$ and
$y \in \Down(f)$, we can apply Theorem~\ref{thm:tarski2} to prove that
$x \le p \le y$, and therefore $x \le y$. This rules out the
existence of points $x, y \in L^k$ with $x \not\le y$ satisfying 
$x \in \Up(f)$ and $y \in \Down(f)$.

For the other direction, suppose there are two distinct fixed points $u, v \in L^k$.
Since $u$ and $v$ are distinct, we cannot have both $u \le v$ and $v \le u$. 
Thus we can either take $x$ as $u$ and $y$ as $v$, or vice versa, so that
$x \not \le y$. 
Since $u$ and $v$ are fixed, both are in $\Up(f) \cap \Down(f)$,
and in particular we will have $x \in \Up(f)$ and $y \in \Down(f)$,
which completes the proof.
\end{proof}

\paragraph{\bf The \sutarski problem.}

Our results will focus on a special case of the \tarski problem,
in which all \emph{slices} have unique solutions.
A slice of the lattice $L^k$ is defined by a tuple $s = (s_1, s_2, \dots,
s_k)$, where each element $s_i \in \{1, 2, \dots, n\} \cup \{ \blank \}$. The idea is that, if $s_i \ne
\blank$, then we fix dimension $i$ of $L^k$ to be $s_i$, and if $s_i = \blank$,
then we allow dimension $i$ of $L^k$ to be free. Formally, we define the sliced
lattice $L^k_s = \{ x \in L^k \; : \; x_i = s_i \text{ whenever } s_i \ne \blank
\}$.

Given a slice $s$, and a function $f : L^k \rightarrow L^k$, we define $f_s : L^k_s
\rightarrow L^k_s$ to be the \emph{restriction} of $f$ to the slice $s$. Specifically, for
each $x \in L^k_s$, we define 
\begin{equation*}
\left(f_s(x)\right)_i = \begin{cases}
f_i(x) & \text{if $s_i = \blank$,} \\
s_i & \text{otherwise.} 
\end{cases}
\end{equation*}

The following lemma shows that monotonicity is hereditary across slices.
\begin{lemma}
\label{lem:slice-op}
If $f$ is monotone, then $f_s$ is also monotone for every slice $s$.
\end{lemma}
\begin{proof}
Let $x \le y$ be two points in $L^k_s$. For each dimension $i$ with $s_i =
\blank$ we have 
\begin{equation*}
f_s(x)_i = f(x)_i \le f(y)_i = f_s(y)_i,
\end{equation*}
where the first and third equalities hold by definition, and the middle
inequality holds due to the fact that $f$ is monotone.
For each dimension $i$ with $s_i \ne \blank$ we have 
\begin{equation*}
f_s(x)_i = x_i \le y_i = f_s(y)_i,
\end{equation*}
where again the first and third equalities hold by definition, and the middle
inequality holds due to the fact that $x \le y$. Hence we have shown that
$f_s(x) \le f_s(y)$.
\end{proof}
If $f$ is monotone, then this lemma says that we can apply Tarski's
theorem to each slice's function $f_s$, meaning that $f_s$ must also have at
least one fixed point. However, the property of having a unique fixed point 
is not hereditary: even if $f$ has a unique fixed point, there
may still be slices $s$ such that $f_s$ has many fixed points.

In our containment results, we need the property that all slices have a
unique fixed point, and for that reason we introduce the \sutarski problem, which
requires that \emph{all} slices should have unique fixed points.

\begin{definition}[\sutarski]
Given a function $f : L^k \rightarrow L^k$, find one of the following.
\begin{enumerate}[label=(UT), wide=0pt]
\item \solnlabel{UT} A point $x \in L^k$ such that $f(x) = x$.
\end{enumerate}
\begin{enumerate}[label=(UTV\arabic*), wide=0pt]
\item \solnlabel{UTV1} Two points $x, y \in L^k$ such that $x \le y$ and $f(x) \not\le
f(y)$.
\item \solnlabel{UTV2} A slice $s$ and two points $x, y \in L^k_s$ such that $x
\not\le y$ and $x \in \Up(f_s)$ while $y \in \Down(f_s)$.
\end{enumerate}
\end{definition}
Here we have split the solution types into the \emph{proper} solution
\solnref{UT} which asks for a fixed point of $f$, and \emph{violation} solutions \solnref{UTV1} and \solnref{UTV2}.
Solutions of type \solnref{UTV1} ask for a violation of monotonicity in
exactly the same way as in the \tarski problem.
Solutions of type \solnref{UTV2} use the characterisation of
Lemma~\ref{lem:non-unique} to give a succinct witness that $f_s$ does not have a
unique fixed point for some slice $s$. Note that $(\blank, \blank, \dots,
\blank)$ is also a slice, so a violation-free instance will also have a unique
fixed point.

\section{\sutarski is in UEOPL}

To show that \sutarski is in UEOPL, we will provide a polynomial time reduction
from \sutarski to the \emph{one-permutation discrete contraction} 
(\OPDC) problem, which is defined as follows. This problem is also defined
over a grid of points, so we will reuse the notation $L^k$ to refer to the set
$\{1,2,\dots,n\}^k$. The \OPDC problem uses the concept of an \emph{$i$-slice},
which is a slice $s$ such that $s_j = \blank$ for all $j \le i$. That is, an
$i$-slice only fixes all coordinates after index $i$, and leaves all
coordinates whose indexes are less than or equal to $i$ as free dimensions.

\begin{definition}[\OPDC (\cite{FGMS20})]
\label{def:dc}
Given Boolean circuits $(D_i(p))_{i=1,\dots,k}$,
where each circuit defines a function $D_i : L^k \rightarrow \{\up, \down, \zero\}$, find one of the following
\begin{enumerate}[label=(O1),wide=0pt]
\item \solnlabel{O1} A point $x \in L^k$ such that $D_i(x) = \zero$ for all $i$.
\end{enumerate}
\begin{enumerate}[label=(OV\arabic*),wide=0pt,leftmargin=\parindent]
\item \solnlabel{OV1} An $i$-slice $s$ and two points $x, y \in L^k_s$ with $x \ne y$
such that $D_i(x) = D_i(y) = \zero$ for all $i$ with $s_i = \blank$.
\item \solnlabel{OV2} An $i$-slice $s$, two points $x, y \in L^k_s$, and an index $i$ such that
\begin{itemize}
\item $D_j(x) = D_j(y) = \zero$ for all $j \ne i$, 
\item $y_i = x_i + 1$, and
\item $D_i(x) = \up$ and $D_i(y) = \down$.
\end{itemize}
\item \solnlabel{OV3} An $i$-slice $s$, a point $x \in L^k_s$, and an index $i$ such that
\begin{itemize}
\item $D_j(x) = \zero$ for all $j \ne i$, and either
\item $x_i = 1$ and $D_i(x) = \down$, or
\item $x_i = n$ and $D_i(x) = \up$.
\end{itemize}
\end{enumerate}
\end{definition}


The OPDC problem has been shown to be complete for the complexity class UEOPL,
though we will only need the containment of OPDC in UEOPL for our result.
Moreover, this has been shown for \emph{promise-preserving} reductions. 
If problems A and B both have proper solutions and violation solutions, then a
reduction from A to B is promise-preserving if, when it is applied to an
instance of A that has no violations, it produces an instance of B
that has no violations. For \OPDC, solutions of type \solnref{O1} are the
proper solutions, while solutions of type \solnref{OV1}, \solnref{OV2}, and
\solnref{OV3} are the violations.

\begin{lemma}[\cite{FGMS20}]
\OPDC is UEOPL-complete under promise-preserving reductions.
\end{lemma}

\paragraph{\bf The reduction.}

We now give a polynomial-time promise-preserving reduction from \sutarski to
\OPDC.
Let $\mathcal{U}$ be an instance of \sutarski defined by
$f : L^k \rightarrow L^k$. We construct an instance $\mathcal{O}$ of
\OPDC over the same lattice $L^k$ in the following way. For each index $i$ and point
$x \in L^k$ we define
\begin{equation*}
D_i(x) = \begin{cases}
\up & \text{if $x_i < f(x)_i$,} \\
\zero &  \text{if $x_i = f(x)_i$,} \\
\down & \text{if $x_i > f(x)_i$.} 
\end{cases}
\end{equation*}
This means that $D_i$ captures, for each point $x$, whether $f(x)$ moves up,
down, or not at all in dimension $i$. Given a circuit representing $f$, we can
clearly produce the circuits $D_i$ in polynomial time. The following lemma
implies that the reduction is correct.

\begin{lemma}
Every solution of $\mathcal{O}$ can be mapped in polynomial time to a solution of $\mathcal{U}$. 
Furthermore, proper solutions of $\mathcal{O}$ are mapped to
proper solutions of $\mathcal{U}$, while violation solutions of $\mathcal{O}$ are
mapped to violation solutions of $\mathcal{U}$.
\end{lemma}
\begin{proof}
We enumerate the four possible solutions to the \OPDC instance.
\begin{enumerate}
\item Solutions of type \solnref{O1} give us a point $x \in L^k$ such that $D_i(x)
= \zero$ for all $i$. By definition, this implies that $x_i = f_i(x)$ for all
$i$, which implies that $x = f(x)$. So $x$ is a fixed point of $f$, and thus a
solution of type \solnref{UT} of $\mathcal{U}$.

\item
Solutions of type \solnref{OV1} give us a slice $s$ and two points $x, y \in
L^k_s$ with $x \ne y$ such that $D_i(x) = D_i(y) = \zero$ for all $i$ with $s_i =
\blank$. Using the same reasoning as the first case, we can conclude that $x$
and $y$ are two distinct fixed points of $f_s$. All fixed points of $f_s$ are
simultaneously members of both $\Up(f_s)$ and $\Down(f_s)$ by definition, and
since $x$ and $y$ are distinct, it cannot be the case that we have both $x \le y$ and $y
\le x$. Thus $x$ and $y$ are a solution of type \solnref{UTV2} of
$\mathcal{U}$. 

\item 
Solutions of type \solnref{OV2}
give us a slice $s$, two points $x, y \in L^k_s$
satisfying $y_i = x_i + 1$, and an index $i$ such that
$D_j(x) = D_j(y) = \zero$ for all $j \ne i$, and
$D_i(x) = \up$ while $D_i(y) = \down$.
The following algorithm can produce a solution of $\mathcal{U}$.
\begin{enumerate}
\item First we observe the following inequalities, which follow by definition.
We have that $x_j = f(x)_j$ for all $j \ne i$, and $x_i < f(x)_i$ for index $i$,
which implies that $x < f_s(x)$. A symmetric argument likewise proves that
$f_s(y) < y$. 

\item 
Since $x < f_s(x)$ we either have that $f_s(x) \le f_s(f_s(x))$, or $x$ and
$f_s(x)$ give us a violation of monotonicity and thus a solution of type
\solnref{UTV1}, where we note that two points that violate monotonicity
for $f_s$ also violate monotonicity for $f$. 
Applying the same argument symmetrically to $y$ allows us to either obtain a
violation of type \solnref{UTV1}
or conclude that $f_s(f_s(y)) \le f_s(y)$.

\item
So assuming that we are not already done, we have that $f_s(x) \in \Up(f_s)$ and
$f_s(y) \in \Down(f_s)$. Observe that for index $i$ we have $f_s(y) \le y_i - 1 = x_i
< f_s(x)$, and so it must be the case that $f_s(x)
\not\le f_s(y)$. Hence $f_s(x)$ and $f_s(y)$ are a solution of type
\solnref{UTV2}.
\end{enumerate}
This algorithm runs in polynomial time, and always produces a violation solution
of $\mathcal{U}$.

\item We argue that solutions of type \solnref{OV3} are impossible.
To see why, note that if $x \in L^k$ is a point, and $x_i = 1$, then $f(x)_i \ge
1$ since $f(x)$ also lies in $L^k$. Hence we cannot have $D_i(x) = \down$.
Likewise if $x_i = n$ then $f(x)_i \le n$, and so we cannot have $D_i(x) = \up$.
Hence, no solutions of type \solnref{OV3} exist.

\end{enumerate}
Note that in the mapping above we have that and all proper solutions of
$\mathcal{O}$ are mapped to proper solutions of $\mathcal{U}$, while violation
solutions $\mathcal{O}$ are mapped to violation solutions of $\mathcal{U}$.
\end{proof}

We remark that nowhere in the lemma above did we use the fact that $s$ was an
$i$-slice. So our reduction would also work if we were to reduce to an
\emph{all-permutation} version  of \OPDC, in which each of the violations were
defined to hold for any slice $s$, rather than just an $i$-slice $s$.

The lemma above, combined with the containment of \OPDC in UEOPL, implies that we have shown the following theorem.

\begin{theorem}
\sutarski is in UEOPL under promise-preserving reductions.
\end{theorem}

\section{Acknowledgement}

We thank Takashi Ishizuka for spotting that our original proof of Lemma~\ref{lem:non-unique} was
incomplete.

\bibliographystyle{abbrv}
\bibliography{references}

\end{document}